\documentclass[11pt]{article}
\usepackage{amsmath}
\usepackage{amsfonts}
\usepackage{amsthm}
\usepackage{amsmath,amssymb}
\usepackage{graphicx}
\usepackage{color}
\usepackage{url}
\usepackage{graphicx}
\usepackage{appendix}
\usepackage{booktabs}

\renewcommand{\arraystretch}{1.2}
\allowdisplaybreaks[4]

\newtheorem{theorem}{Theorem}
\newtheorem{lemma}{Lemma}

\newtheorem{definition}{Definition}

\newcommand{\RNum}[1]{\uppercase\expandafter{\romannumeral #1\relax}}

\newcommand{\gf}{{\mathbb{F}}}

\newcommand{\ls}[1]
    {\dimen0=\fontdimen6\the\font\lineskip=#1\dimen0
     \advance\lineskip.5\fontdimen5\the\font
     \advance\lineskip-\dimen0
     \lineskiplimit=0.9\lineskip
     \baselineskip=\lineskip
     \advance\baselineskip\dimen0
     \normallineskip\lineskip\normallineskiplimit\lineskiplimit
     \normalbaselineskip\baselineskip
     \ignorespaces}

\begin{document}

\bibliographystyle{abbrv}

\title{On the second-order zero differential spectra of some power functions over finite fields}

\author{Yuying Man\footnotemark[1]\thanks{Y. Man and X. Zeng are with Hubei Key Laboratory of Applied Mathematics, Faculty of Mathematics and Statistics, Hubei University, Wuhan 430062, China. Email:yuying.man@aliyun, xzeng@hubu.edu.cn},
Nian Li\footnotemark[2]\thanks{N. Li and Z. Xiang are with Hubei Key Laboratory of Applied Mathematics, School of Cyber Science and Technology, Hubei University, Wuhan 430062, China. Email: nian.li@hubu.edu.cn, xiangzejun@hubu.edu.cn},
 Zejun Xiang\footnotemark[2], Xiangyong Zeng\footnotemark[1]
}
\date{\today}
\maketitle

\thispagestyle{plain} \setcounter{page}{1}

\begin{abstract}
Boukerrou et al. (IACR Trans. Symmetric Cryptol. 2020(1), 331-362) introduced the notion of Feistel Boomerang Connectivity Table (FBCT), the Feistel counterpart of the Boomerang Connectivity Table (BCT), and the Feistel boomerang uniformity (which is the same as the second-order zero differential uniformity in even characteristic). FBCT is a crucial table for the analysis of the resistance of block ciphers to power attacks such as differential and boomerang attacks. It is worth noting that the coefficients of FBCT are related to the second-order zero differential spectra of functions. In this paper, by carrying out certain finer manipulations of solving specific equations over the finite field $\mathbb{F}_{p^n}$, we explicitly determine the second-order zero differential spectra of some power functions with low differential uniformity, and show that our considered functions also have low second-order zero differential uniformity. Our study pushes further former investigations on second-order zero differential uniformity and Feistel boomerang differential uniformity for a power function $F$.

\noindent{\bf Keywords} Feistel Boomerang Connectivity Table, Feistel boomerang
differential uniformity, Second-order zero differential spectra, Second-order zero differential uniformity

\noindent{\bf MSC (2020)} 94A60, 11T06

\end{abstract}

\section{Introduction}

Boomerang attack, introduced by Wagner \cite{BA} in 1999, is a crucial cryptanalytical technique on block cyphers. To analyze the boomerang attack of block cyphers in a better way, analogous to the Difference Distribution Table (DDT) concerning the differential attack, in Eurocrypt 2018, Cid et al. in \cite{BCT} introduced a new tool known as BCT to measure the resistance of an S-box against boomerang attacks.
To consider the case of ciphers following a Feistel Network structure,
Boukerrou et al. \cite{def-FBCT} introduced the notion of FBCT, as an extension for Feistel cipher, where the employed S-boxes may not be permutations.
They also studied the properties of the FBCT of $F$ over the finite field of even characteristic, and showed that $F$ is an almost perfect nonlinear (APN) function (which have the lowest differential uniformity over even characteristic finite fields) if and only if FBCT of $F$ is $0$ for $a,b\in \mathbb{F}_{2^n}$ with $ab(a+b)\neq 0$. Furthermore, Garg et al. \cite{GHRS-SOZ} showed that, for odd characteristic, if $F$ is second-order zero differentially 1-uniform then it has to be an APN function.

In \cite{CC-SOZ}, the authors studied the second-order zero differential spectra of the inverse function and some APN functions in odd characteristic and they also showed that these functions have low second-order zero differential uniformity. Eddahmani et al. \cite{ex-FBCT} investigated the FBCT of the inverse, the Gold and the Bracken-Leander functions over $\mathbb{F}_{2^n}$. They further determined the Feistel boomerang differential uniformity of these functions. The authors in \cite{FBCT-MML} provided explicit values of all entries in the FBCT of a specific power function and they determined the Feistel boomerang differential uniformity of this function. Recently, Garg et al.  computed the second-order zero differential spectra of several APN and other low differential uniform functions in \cite{GHRS-SOZ} and \cite{arxiv-SOZ}. They also given the second-order zero differential uniformity of these functions.
Table 1 gives the known power functions with the second-order zero differential uniformity over finite fields. In this paper, in order to in-depth analysis of $F$ over finite fields concerning their cryptographic properties, we studied the second-order zero differential spectra of some power functions with low differential uniformity in finite fields. In addition, these mappings considered in this paper also have low second-order zero differential uniformity.

\begin{table}
\caption{Power functions $F(x)=x^d$ over $\gf_{p^n}$ with known second-order differential uniformity}\label{table-1}
\renewcommand\arraystretch{1.2}
\setlength\tabcolsep{14pt}
\centering
\begin{tabular}{ccccc}
		\toprule
     $p$ & $d$ & Condition & $\nabla_F$ &  Ref. \\
       \midrule
       $p=2$ &  $2^n-2$ & $n$ odd or $n$ even & 2 or 4  &  \cite{ex-FBCT} \\
       $p=2$ & $2^k+1$ & ${\rm gcd}(n,k)=d$ & $2^d$ & \cite{ex-FBCT}\\
       $p=2$ &$2^{2k}+2^k+1$ & $n=4k$ & $2^{2k}$ & \cite{ex-FBCT}\\
       $p=2$ & $2^{m+1}-1$ & $n=2m+1$ or $n=2m$ & 2 or $2^m$ & \cite{FBCT-MML} \\
       $p=2$ &  $2^m-1$ & $n=2m+1$ or $n=2m$ & $2^m-4$ & \cite{arxiv-SOZ} \\
       $p=2$ & $21$ & $n$ odd or $n$ even & 4 or 16 & \cite{GHRS-SOZ} \\
       $p=2$ & $2^n-2^s$ & ${\rm gcd}(n, s+1)=1$, $n-s=3$ & 4 & \cite{GHRS-SOZ} \\
      $p>3$ &  $3$ & any & $1$ & \cite{CC-SOZ} \\
      $p=3$ & $3^n-3$ & $n>1$ is odd & $2$ & \cite{CC-SOZ} \\
      $p>3$ & $p^n-2$ & $p^n\equiv 2\,\, ({\rm mod}\,\,3)$ & $1$ & \cite{CC-SOZ} \\
      $p>3$ & $p^n-2$ & $p^n\equiv 1\,\, ({\rm mod}\,\,3)$ & $3$ & \cite{CC-SOZ} \\
      $p=3$ & $3^n-2$ & any & $3$ & \cite{CC-SOZ} \\
      $p>3$ & $p^m+2$ & $n=2m$, $p^m\equiv 1\,\, ({\rm mod}\,\,3)$ & $1$ & \cite{CC-SOZ} \\
      $p>3$ & $4$ & $n>1$ & $2$ & \cite{arxiv-SOZ} \\
      $p$ & $\frac{2p^n-1}{3}$ & $p^n\equiv 2\,\, ({\rm mod}\,\,3)$ & $1$ & \cite{arxiv-SOZ} \\
      $p>3$ & $\frac{p^k+1}{2}$ & ${\rm gcd}(2n, k)=1$ & $\frac{p-3}{2}$ & \cite{arxiv-SOZ}\\
      $p=3$ & $\frac{3^n-1}{2}+2$ & $n$ odd & $3$ & \cite{arxiv-SOZ} \\
      $p=3$ &  $2\cdot 3^{\frac{n-1}{2}}+1$ & any & $3$ & \cite{GHRS-SOZ} \\
      $p$ & $\frac{p^n+1}{4}+\frac{p^n-1}{2}$ & $p^n\equiv 3\,\, ({\rm mod}\,\,8)$ & $8$ or $18$ & \cite{GHRS-SOZ} \\
      $p$ & $\frac{p^n+1}{4}$ & $p^n\equiv 7\,\, ({\rm mod}\,\,8)$ & $8$ or $18$ & \cite{GHRS-SOZ} \\
      $p=2$ & $7$ & any & $4$ & This paper \\
      $p=2$ & $2^{m+1}+3$ & $n=2m+1$ or $n=2m$ & $4$ or $2^m$ & This paper \\
      $p>2$ &  $5$ & any & $3$ & This paper \\
      $p=3$ &  $7$ & any & $3$ & This paper \\
\bottomrule
\end{tabular}
\end{table}

The remainder of this paper is organized as follows. In Section \ref{pre}, we present some basic notations and a few known helpful results in the technical part of the paper. In Section \ref{even-result1}, we consider the second-order zero differential spectra of two classes of power functions in even characteristic. Section \ref{odd-result1} studies the second-order zero differential spectra of two classes of power functions in odd characteristic. Section \ref{con-remarks} concludes this paper.

\section{Preliminaries}\label{pre}

Throughout this paper, $\mathbb{F}_{p^n}$ denotes the finite field with $p^n$ elements and ${\rm Tr}_{m}^n(x)=x+x^{p^m}+x^{p^{2m}}+\cdots +x^{p^{n-m}}$ denotes the trace function from $\mathbb{F}_{p^n}$ to $\mathbb{F}_{p^m}$, where $m$, $n$ are positive integers and $m|n$.

In this section,  we recall some basic definitions and present some results which will be used frequently in this paper.

 \begin{definition}\label{definition-DDT}\rm (\cite{def-DDT})
 Let $F(x)$ be a  mapping from $\mathbb{F}_{p^n}$ to itself. The Difference Distribution Table (DDT) of $F(x)$ is a $p^n \times p^n$ table where the entry at $(a,b)\in \mathbb{F}_{p^n}^2$ is defined by
$${\rm DDT}_{F}(a,b)=|\{x \in \gf_{p^n}: F(x+a)-F(x)=b \}|.$$

The mapping $F(x)$ is said to be \emph{differentially $\delta$-uniform} if $\Delta_F$=$\delta$ \cite{def-DDT}, and accordingly $\Delta_F$ is called the \emph{differential uniformity} of $F(x)$, where
$$
\Delta_F=\max _{a,b \in \mathbb{F}_{p^n},a\ne 0} {\rm DDT}_{F}(a,b).
$$
When $F(x)$ is used as an S-box inside a cryptosystem, the smaller the value $\Delta_F$ is, the better the contribution of $F(x)$ to the resistance against differential attack.
\end{definition}

The definitions of the second-order zero differential spectrum and the FBCT of $F(x)$ are given as follows.

\begin{definition}\label{definition-SOZ}\rm (\cite{def-FBCT})
Let $F(x)$ be a mapping from $\mathbb{F}_{p^n}$ to itself. The second-order zero differential spectrum with respect to $a$, $b$ of $F$ is defined as
$$\nabla_{F}(a,b)=|\{x \in \gf_{p^n}: F(x+a+b)-F(x+b)-F(x+a)+F(x)=0 \}|.$$
The second-order zero differential uniformity of $F(x)$ is defined by $\nabla_{F}=\max\{\nabla_{F}(a,b): a\ne b, a, b \in \mathbb{F}_{2^n}\setminus\{0\}\}$ for $p=2$ and $\nabla_{F}=\max\{\nabla_{F}(a,b): a, b \in \mathbb{F}_{p^n}\setminus\{0\}\}$ for $p>2$.
The mapping $F(x)$ is said to be second-order zero differential $k$-uniform if $\nabla_{F}=k$.
\end{definition}

\begin{definition}\label{definition-FBCT}\rm (\cite{def-FBCT})
Let $F(x)$ be a mapping from $\mathbb{F}_{2^n}$ to itself. The Feistel Boomerang Connectivity Table (FBCT) is a $2^n\times 2^n$ table defined for $(a,b)\in \mathbb{F}_{2^n}^2$ by
$${\rm FBCT}_{F}(a,b)=|\{x \in \gf_{2^n}: F(x)+F(x+a)+F(x+b)+F(x+a+b)=0 \}|.$$
Clearly, the FBCT satisfies ${\rm FBCT}_{F}(a,b)=2^n$ if $ab(a+b)=0$. Hence, the Feistel boomerang uniformity of $F(x)$ is defined by
$$
\beta(F)=\max _{a, b \in \mathbb{F}_{2^n}, ab(a+b)\ne 0} {\rm FBCT}_{F}(a,b).
$$
\end{definition}

The basic properties of the FBCT are studied in~\cite{def-FBCT}. Typically, the FBCT satisfies the following properties.
\begin{itemize}
\item Symmetry: ${\rm FBCT}_F(a,b)={\rm FBCT}_F(b,a)$ for all $a,b\in \mathbb{F}_{2^n}$.
\item Multiplicity: ${\rm FBCT}_F(a,b)\equiv 0\pmod 4$ for all $a,b\in \mathbb{F}_{2^n}$.
\item First line: ${\rm FBCT}_F(0,b)=2^n$ for all $b\in \mathbb{F}_{2^n}$.
\item First column: ${\rm FBCT}_F(a,0)=2^n$ for all $a\in \mathbb{F}_{2^n}$.
\item Diagonal: ${\rm FBCT}_F(a,a)=2^n$ for all $a\in \mathbb{F}_{2^n}$.
\item Equalities: ${\rm FBCT}_F(a,a)={\rm FBCT}_F(a,a+b)$ for all $a,b\in \mathbb{F}_{2^n}$.
\end{itemize}

We recall the following lemma which concerns quadratic equations in $\gf_{2^n}$.

\begin{lemma}\label{lemma1-qua-root}\rm (\cite{def-squ})
Let $a, b, c \in \mathbb{F}_{2^n}$, $a\ne 0$ and $F(x)=ax^2+bx+c$. Then
\begin{itemize}
\item [\rm (i)] $F(x)$ has exactly one root in $\mathbb{F}_{2^n}$ if and only if $b=0$;
\item [\rm (ii)] $F(x)$ has exactly two roots in $\mathbb{F}_{2^n}$ if and only if $b\ne 0$ and ${\rm Tr}_{1}^n(\frac{ac}{b^2})=0$,
\item [\rm (iii)] $F(x)$ has no root in $\mathbb{F}_{2^n}$ if and only if $b\ne 0$ and ${\rm Tr}_{1}^n(\frac{ac}{b^2})=1$.
\end{itemize}
\end{lemma}

The following lemma described a method to solve $F(z)=z^{p^k}+z+B$ over $\gf_{2^n}$.

\begin{lemma}\label{lemma1-a^k=1-root}\rm (\cite{ex-FBCT})
Let $n$ and $k$ be positive integers such that $k<n$. Let $d={\rm gcd}(k, n)$, $l=n/d>1$, and $\beta_{l-1}={\rm Tr}_d^n(b)$. Then, the trinomial  $F(z)=z^{p^k}+z+B$ has no root if $\beta_{l-1}\ne 0$, and has $2^d$ roots $x+\delta\tau$ in $\gf_{2^n}$ if $\beta_{l-1}=0$ where
$\delta\in \gf_{2^d}$ and $\tau \in \gf_{2^n}$ with any element satisfying $\tau^{2^k-1}=1$, and
$$
x=\frac{1}{{\rm Tr}_d^n(c)}\sum\limits_{i=0}^{l-1}\Big(\sum\limits_{j=0}^ic^{2^{kj}}\Big)B^{2^{ki}},
$$
with any $c\in \gf_{2^n}^*$ satisfying ${\rm Tr}_d^n(c)\in \gf_{2^d}^*$.
\end{lemma}

An interesting result on quartics equations is given as below.
\begin{lemma}\label{lemma1-quar-root}\rm (\cite{quar-root})
Let $F(x)=x^4+a_2x^2+a_1x+a_0$ with $a_0a_1\ne 0$ and the companion cubic $G(y)=y^3+a_2y+a_1$ with the roots $r_1$, $r_2$, $r_3$. When the roots exist in $\gf_{2^n}$, set $\omega_i=(a_0r_i^2)/a_1^2$. Let a polynomial $h$ as $h=(1, 2, 3, \cdots)$ over some field to mean that it decomposes as a product of degree $1$, $2$, $3$, $\cdots$, over that field. The factorization of $F(x)$ over $\gf_{2^n}$ is characterized as follows:
\begin{itemize}
\item [\rm (i)] $F=(1, 1, 1, 1)\Leftrightarrow G=(1, 1, 1)$ and ${\rm Tr}_1^n(\omega_1)={\rm Tr}_1^n(\omega_2)={\rm Tr}_1^n(\omega_3)=0$;
\item [\rm (ii)] $F=(2, 2)\Leftrightarrow G=(1, 1, 1)$ and ${\rm Tr}_1^n(\omega_1)=0$, ${\rm Tr}_1^n(\omega_2)={\rm Tr}_1^n(\omega_3)=1$;
\item [\rm (iii)] $F=(1, 3)\Leftrightarrow G=(3)$;
\item [\rm (iv)] $F=(1, 1, 2)\Leftrightarrow G=(1, 2)$ and ${\rm Tr}_1^n(\omega_1)=0$;
\item [\rm (v)] $F=(4)\Leftrightarrow G=(1, 2)$ and ${\rm Tr}_1^n(\omega_1)=1$.
\end{itemize}
\end{lemma}

\section{The second-order zero differential spectra of functions in even characteristic}\label{even-result1}
This section is devoted to presenting a detailed study of the FBCT of two power mappings over $\gf_{2^n}$. The main results are given by the following theorems, which are derived through the computation of the number of solutions over $\mathbb{F}_{2^n}$ of the equation $F(x+a+b)+F(x+b)+F(x+a)+F(x)=0$.

Blondeau et al. \cite{BCC2011IT} determined the differential spectrum of $F(x)=x^7$ by means of the values of some Kloosterman sums and they showed that $F(x)$ is differentially $6$-uniform function over $\gf_{2^n}$ with $n\geq 4$ (where $n$ is odd or $n$ is even). In the following theorem, we compute its second-order zero differential spectrum of $F(x)$ over $\gf_{2^n}$,  where $n$ is odd or $n$ is even.

\begin{theorem}\label{the-FBCT-7}
Let $F(x)=x^7$ be a power mapping over $\gf_{2^n}$. For $a, b \in \gf_{2^n}$, let $c=\frac{a}{b}$, $a_0=(c^2+c+1)^2$, $a_1=c^2+c$, $a_2=c^2+c+1$, $\omega_1=\frac{a_0}{a_1^2}$, $\omega_2=\frac{a_0c^2}{a_1^2}$, $\omega_3=\frac{a_0(c+1)^2}{a_1^2}$. Then, $F(x)$ is second-order zero differential $4$-uniform. Moreover, when $n$ is even,
$$
{\nabla}_{F}(a,b)=
\begin{cases}
   2^n, &  {\rm if}\,\, ab(a+b)=0;\\
   4, &  {\rm if}\,\, ab(a+b)\ne 0,\,\, a\in \{b\omega, b(\omega+1)\}\\
      &   {\rm or}\,\, ab(a+b)\ne 0,\,\,  {\rm Tr_1^n}(\omega_1)={\rm Tr_1^n}(\omega_2)={\rm Tr_1^n}(\omega_3)=0;\\
   0, &  {\rm otherwise},\\
\end{cases}
$$
where $\omega$ and $\omega+1$ be the solution of $c^2+c+1$ in $\gf_{2^n}$.

When $n$ is odd,
$$
{\nabla}_{F}(a,b)=
\begin{cases}
   2^n, &  {\rm if}\,\, ab(a+b)=0;\\
   4, &  {\rm if}\,\, ab(a+b)\ne 0,\,\, {\rm Tr_1^n}(\omega_1)={\rm Tr_1^n}(\omega_2)={\rm Tr_1^n}(\omega_3)=0;\\
   0, &  {\rm otherwise}.\\
\end{cases}
$$
\end{theorem}

\begin{proof}
To prove this theorem, according to Definition \ref{definition-FBCT}, we need to
count the number of the solutions of
\begin{equation}\label{FBCT-d=7}
x^7+(x+a)^7+(x+b)^7+(x+a+b)^7=0,
\end{equation}
where $a,b\in \gf_{2^n}$.

When $a=0$ or $b=0$ or $a=b$ with $a\ne 0$, it can be easily seen that (\ref{FBCT-d=7}) holds for all $x\in \gf_{2^n}$, which gives
$$
 {\rm FBCT}_{F}(a,b)=2^n.
$$
Assume that $ab(a+b)\ne 0$.  Let $c=\frac{a}{b}$ and $y=\frac{x}{b}$, where $c\ne 0,1$.
Then, (\ref{FBCT-d=7}) is equivalent to
$$
b^7(y^7+(y+c)^7+(y+1)^7+(y+c+1)^7)=0.
$$
Since $b\ne 0$, thus we only need to consider the solutions of
\begin{equation}\label{FBCT-d=7-0}
y^7+(y+c)^7+(y+1)^7+(y+c+1)^7=0.
\end{equation}
If $y=0, 1, c, c+1$, then the above equation can be reduce to
$$
(c+1)^7+c^7+1=(c^2+c)(c^2+c+1)^2=0,
$$
we have $c^2+c+1=0$ since $c\ne 0, 1$. By Lemma \ref{lemma1-qua-root}, $c^2+c+1=0$ has no solution in $\gf_{2^n}$ when $n$ is odd, and it has two solutions in $\gf_{2^n}$ when $n$ is even. Let $\omega$ and $\omega+1$ be the solution of $c^2+c+1=0$, that is $c\in \{\omega, \omega+1\}$, then $y=0, 1, c, c+1$ are the solutions of (\ref{FBCT-d=7-0}).

Next ,we assume that $y\ne 0, 1, c, c+1$.
Expanding each of the terms of the above equation leads to
\begin{equation}\label{FBCT-d=7-3}
y^4(c^2+c)+y^2(c^4+c)+y(c^4+c^2)+(c^2+c)(c^2+c+1)^2=0.
\end{equation}
Since $c\ne 0,1$, (\ref{FBCT-d=7-3}) is equivalent to
\begin{equation}\label{FBCT-d=7-4}
y^4+(c^2+c+1)y^2+(c^2+c)y+(c^2+c+1)^2=0.
\end{equation}
When $n$ is odd, $c^2+c+1=0$ has no solution in $\gf_{2^n}$. When $n$ is even, if $c^2+c+1=0$, then (\ref{FBCT-d=7-4}) can be reduced to $y^4+(c^2+c)y=0$, it can be easily seen that the above equation has four solutions $y=0, 1, c, c+1$, which contradicts with $y\ne 0, 1, c, c+1$. Then we have $c^2+c+1\ne 0$ for $n$ is odd or $n$ is even.
By Lemma \ref{lemma1-quar-root}, the companion cubic polynomial of (\ref{FBCT-d=7-4}) is
\begin{equation*}
G(z)=z^3+(c^2+c+1)z+(c^2+c),
\end{equation*}
which can be factored as $(z+1)(z+c)(z+c+1)$ in $\gf_{2^n}$. If $G(z)=0$, we get $z_1=1$, $z_2=c$ and $z_3=c+1$. Let $a_0=(c^2+c+1)^2$, $a_1=c^2+c$, $a_2=c^2+c+1$, $\omega_1=\frac{a_0z_1^2}{a_1^2}=\frac{a_0}{a_1^2}$, $\omega_2=\frac{a_0z_2^2}{a_1^2}=\frac{a_0c^2}{a_1^2}$, $\omega_3=\frac{a_0z_3^2}{a_1^2}=\frac{a_0(c+1)^2}{a_1^2}$. Since $G(z)$ can be factored as $(1, 1, 1)$, from Lemma \ref{lemma1-quar-root}, we can easily seen that (\ref{FBCT-d=7-4}) has four solutions in $\gf_{2^n}$ if and only if  ${\rm Tr_1^n}(\omega_1)={\rm Tr_1^n}(\omega_2)={\rm Tr_1^n}(\omega_3)=0$. This completes the proof.
\end{proof}

Let $m\geq 5$ be an odd integer. For $d=2^{m+1}+3$, Blondeau et al. \cite{2010-IJICT} conjectured that the power mapping $F(x)=x^d$ over $\gf_{2^n}$ is differentially $8$-uniform, where $n=2m$. Xiong et al. \cite{2017-DCC} confirmed this conjecture and computed the differential spectrum of $F(x)$. In order to get further cryptographic properties of $F(x)=x^{2^{m+1}+3}$, we determine the second-order zero differential spectrum of $F(x)$ over $\gf_{2^n}$ (where $n=2m$ or $n=2m+1$) in the following theorem.

\begin{theorem}\label{the-FBCT-d=niho}
Let $F(x)=x^{2^{m+1}+3}$ be a power mapping over $\gf_{2^n}$. For $a, b \in \gf_{2^n}$, let $c=\frac{a}{b}$. When $n=2m$, let $a_0=\Big(\frac{c^{2^m}+c^2}{(c^{2^m}+c)}\Big)^4+\frac{(c^{2^{m+1}}+c)(c^4+c)}{(c^{2^m}+c)^2}$, $a_1=c^2+c$, $a_2=c^2+c+1$, $\omega_1=\frac{a_0}{a_1^2}$, $\omega_2=\frac{a_0c^2}{a_1^2}$, $\omega_3=\frac{a_0(c+1)^2}{a_1^2}$. Then, $F(x)$ is second-order zero differential $2^m$-uniform. Moreover,
$$
{\nabla}_{F}(a,b)=
\begin{cases}
  2^n, &  {\rm if}\,\, ab(a+b)=0;\\
   4, &   {\rm if}\,\,  ab(a+b)\ne 0,\,\, a\in \{b\omega, b(\omega+1)\}\\
      &  {\rm or}\,\,  ab(a+b)\ne 0,\,\, \frac{a}{b}\in \gf_{2^n}\backslash \gf_{2^m},\,\,
       {\rm Tr_1^n}(\omega_1)={\rm Tr_1^n}(\omega_2)={\rm Tr_1^n}(\omega_3)=0; \\
   2^m, &  {\rm if}\,\, ab(a+b)\ne 0,\,\, \frac{a}{b}\in \gf_{2^m};\\
   0, &  {\rm otherwise}.\\
\end{cases}
$$
When $n=2m+1$, let $a_0=\frac{(c^{2^m}+c^2)^2(c^{2^{m+1}}+c^2)}{c^{2^{m+2}}+c^2}+\frac{(c^{2^{m+1}}+c)(c^4+c)}{c^{2^{m+1}}+c^2}$, $a_1=c^2+c$, $a_2=c^2+c+1$, $\omega_1=\frac{a_0}{a_1^2}$, $\omega_2=\frac{a_0c^2}{a_1^2}$, $\omega_3=\frac{a_0(c+1)^2}{a_1^2}$. Then, $F(x)$ is second-order zero differential $4$-uniform. Moreover,
$$
{\nabla}_{F}(a,b)=
\begin{cases}
   4, &  {\rm if}\,\,  ab(a+b)\ne 0,\,\, \frac{a}{b}\in \gf_{2^n}\backslash \gf_{2^m},\,\, {\rm Tr_1^n}(\omega_1)={\rm Tr_1^n}(\omega_2)={\rm Tr_1^n}(\omega_3)=0;\\
   2^n, &  {\rm if}\,\, ab(a+b)=0;\\
   0, &  {\rm otherwise}.\\
\end{cases}
$$
\end{theorem}

\begin{proof}
To prove this theorem, according to Definition \ref{definition-FBCT}, we need to
count the number of the solutions of
\begin{equation}\label{FBCT-d=niho}
x^{2^{m+1}+3}+(x+a)^{2^{m+1}+3}+(x+b)^{2^{m+1}+3}+(x+a+b)^{2^{m+1}+3}=0,
\end{equation}
where $a,b\in \gf_{2^n}$.

When $a=0$ or $b=0$ or $a=b$ with $a\ne 0$, it can be easily seen that (\ref{FBCT-d=niho}) holds for all $x\in \gf_{2^n}$, which gives
$$
 {\rm FBCT}_{F}(a,b)=2^n.
$$
Assume that $ab(a+b)\ne 0$. Let $c=\frac{a}{b}$ and $y=\frac{x}{b}$, we have $c\ne 0,1$.
Then, (\ref{FBCT-d=niho}) is equivalent to
$$
b^{2^{m+1}+3}(y^{2^{m+1}+3}+(y+c)^{2^{m+1}+3}+(y+1)^{2^{m+1}+3}+(y+c+1)^{2^{m+1}+3})=0.
$$
Since $b\ne 0$, thus we only need to consider the solutions of
\begin{equation}\label{FBCT-d=niho-0}
y^{2^{m+1}+3}+(y+c)^{2^{m+1}+3}+(y+1)^{2^{m+1}+3}+(y+c+1)^{2^{m+1}+3}=0.
\end{equation}
If $y=0$, $1$, $c$ or $c+1$, then (\ref{FBCT-d=niho-0}) becomes
\begin{equation*}
(c^{2^{m+1}}+c)(c^2+c+1)=0.
\end{equation*}
When $n=2m+1$, then $c^2+c+1$ has no solution in $\gf_{2^n}$, thus we get $c^{2^{m+1}}+c=0$, this implies that $c\in \gf_{2^{m+1}}$. Since ${\rm gcd}(m+1, 2m+1)=1$, we have $c\in \gf_{2}$, which contradicts with $c\ne 0, 1$.  Then we get $y=0$, $1$, $c$ or $c+1$ are not solutions of (\ref{FBCT-d=niho-0}). When $n=2m$, we have that $c^2+c+1$ has two solutions in $\gf_{2^n}$. Let $\omega$ and $\omega+1$ be the solution of the equation $c^2+c+1$ in $\gf_{2^n}$, that is $c\in \{\omega, \omega+1\}$, then $y=0$, $1$, $c$ or $c+1$ are the solutions of (\ref{FBCT-d=niho-0}). Next, assume that $y\ne 0, 1, c, c+1$.
Expanding each of the terms of (\ref{FBCT-d=niho-0}) gives
\begin{equation}\label{FBCT-d=niho-1}
y^{2^{m+1}}(c^2+c)+y^2(c^{2^{m+1}}+c)+y(c^{2^{m+1}}+c^2)+(c^{2^{m+1}}+c)(c^2+c+1)=0.
\end{equation}

We start by considering the case $n=2m$.

{\textbf{Case 1:}} Assume that $c\in \gf_{2^m}$. Then $c^{2^m}=c$ and $c^{2^{m+1}}=c^2$. (\ref{FBCT-d=niho-1}) reduces to
\begin{equation*}
y^{2^{m+1}}(c^2+c)+y^2(c^2+c)+(c^2+c)(c^2+c+1)=0,
\end{equation*}
which can be rewritten as
\begin{equation}\label{FBCT-d=niho-2^m}
y^{2^{m+1}}+y^2+(c^2+c+1)=0,
\end{equation}
since $c\ne 0, 1$.
Let $z=y^2$, (\ref{FBCT-d=niho-2^m}) is equivalent to
\begin{equation}\label{FBCT-d=niho-2^m-1}
z^{2^m}+z+(c^2+c+1)=0.
\end{equation}
By Lemma \ref{lemma1-a^k=1-root}, in our case, we have ${\rm Tr}_d^n(c^2+c+1)=c^2+c+1+ (c^2+c+1)^{2^m}=0$, since $d={\rm gcd}(n,m)=m$, $l=n/d=2$ when $n=2m$. This implies that (\ref{FBCT-d=niho-2^m-1}) has $2^m$ solutions when $c\in \gf_{2^m}$.

{\textbf{Case 2:}} Assume that $c\in \gf_{2^n}\backslash \gf_{2^m}$. Raising $2^m$-th power to (\ref{FBCT-d=niho-1}) leads to
\begin{equation}\label{FBCT-d=niho-2}
y^2(c^{2^{m+1}}+c^{2^m})+y^{2^{m+1}}(c^2+c^{2^m})+y^{2^m}(c^2+c^{2^{m+1}})+(c^2+c^{2^m})(c^{2^{m+1}}+c^{2^m}+1)=0.
\end{equation}
Since $c\ne 0, 1$, then $c^2+c\ne 0$ and $c^2+c^{2^m}\ne 0$.
If $c^2+c^{2^m}=0$, namely, $c^{2^{m+1}}+c=0$, then (\ref{FBCT-d=niho-1}) can be reduced to $y^{2^{m+1}}+y=0$. It can be easily seen that the above equation has four solutions in $\gf_{2^n}$, which can be written as $y=0, 1, c, c+1$, a contradiction. Thus we have $c^2+c^{2^m}\ne 0$. Multiplying $c^2+c^{2^m}$ and $c^2+c$ on both sides of (\ref{FBCT-d=niho-1}) and (\ref{FBCT-d=niho-2}) and then summing up these two equations gives
$$
(c^2+c)(c^2+c^{2^{m+1}})y^{2^m}+(c^{2^m}+c)^3y^2+(c^2+c^{2^m})(c^{2^{m+1}}+c^2)y+(c^{2^m}+c^2)^2(c^{2^m}+c)=0,
$$
Since $c\ne 0,1$ and $c\notin \gf_{2^m}$, then we have $c^2+c\ne 0$ and $c^2+c^{2^{m+1}}\ne 0$. Thus the above equation is equivalent to
$$
y^{2^m}=\Big(\frac{(c^{2^m}+c)^3}{(c^2+c)(c^2+c^{2^{m+1}})}\Big)y^2+\Big(\frac{(c^2+c^{2^m})(c^{2^{m+1}}+c^2)}{(c^2+c)(c^2+c^{2^{m+1}})}\Big)y
+\frac{(c^{2^m}+c^2)^2(c^{2^m}+c)}{(c^2+c)(c^2+c^{2^{m+1}})}.
$$
Squaring both sides of the above equation and substituting it into (\ref{FBCT-d=niho-1}) and then multiplying $(c^2+c)(c^{2^m}+c)^4$ on the both sides of this equation, we have
$$
\begin{aligned}
&(c^{2^m}+c)^6y^4+(c^{2^m}+c)^6(c^2+c+1)y^2+(c^{2^m}+c)^6(c^2+c)y\\
&+(c^{2^m}+c^2)^4(c^{2^m}+c)^2+(c^4+c)(c^{2^m}+c)^4(c^{2^m+1}+c)=0.
\end{aligned}
$$
Since $c^{2^m}+c\ne 0$, then the above equation is equivalent to
\begin{equation}\label{FBCT-d=niho-3}
y^4+(c^2+c+1)y^2+(c^2+c)y+\Big(\frac{c^{2^m}+c^2}{(c^{2^m}+c)}\Big)^4+\frac{(c^{2^{m+1}}+c)(c^4+c)}{(c^{2^m}+c)^2}=0.
\end{equation}
If $\Big(\frac{c^{2^m}+c^2}{(c^{2^m}+c)}\Big)^4+\frac{(c^{2^{m+1}}+c)(c^4+c)}{(c^{2^m}+c)^2}=0$, then (\ref{FBCT-d=niho-3}) can be reduced to
$$y^4+(c^2+c+1)y^2+(c^2+c)y=y(y+1)(y+c)(y+c+1)=0,$$
this implies that (\ref{FBCT-d=niho-3}) has four solutions in $\gf_{2^n}$, namely, $y=0$, $y=1$, $y=c$ and $y=c+1$, which contradicts with $y\ne 0, 1, c, c+1$, then we have $\Big(\frac{c^{2^m}+c^2}{(c^{2^m}+c)}\Big)^4+\frac{(c^{2^{m+1}}+c)(c^4+c)}{(c^{2^m}+c)^2}\ne 0$. By Lemma \ref{lemma1-quar-root}, the companion cubic polynomial of (\ref{FBCT-d=niho-3}) is
\begin{equation*}
G(z)=z^3+(c^2+c+1)z+(c^2+c),
\end{equation*}
which can be factored as $(z+1)(z+c)(z+c+1)$ in $\gf_{2^n}$. If $G(z)=0$, we get $z_1=1$, $z_2=c$ and $z_3=c+1$. Let $a_0=\Big(\frac{c^{2^m}+c^2}{(c^{2^m}+c)}\Big)^4+\frac{(c^{2^{m+1}}+c)(c^4+c)}{(c^{2^m}+c)^2}$, $a_1=c^2+c$, $a_2=c^2+c+1$, $\omega_1=\frac{a_0z_1^2}{a_1^2}=\frac{a_0}{a_1^2}$, $\omega_2=\frac{a_0z_2^2}{a_1^2}=\frac{a_0c^2}{a_1^2}$, $\omega_3=\frac{a_0z_3^2}{a_1^2}=\frac{a_0(c+1)^2}{a_1^2}$. Since $G(z)$ can be factored as $(1, 1, 1)$, from Lemma \ref{lemma1-quar-root}, we can easily seen that (\ref{FBCT-d=niho-3}) has four solutions in $\gf_{2^n}$ if and only if  ${\rm Tr_1^n}(\omega_1)={\rm Tr_1^n}(\omega_2)={\rm Tr_1^n}(\omega_3)=0$.

We now consider the case $n=2m+1$ as follows.

Raising $2^m$-th power to (\ref{FBCT-d=niho-1}) leads to
\begin{equation}\label{FBCT-d=niho-odd-1}
y(c^{2^{m+1}}+c^{2^m})+y^{2^{m+1}}(c+c^{2^m})+y^{2^m}(c+c^{2^{m+1}})+(c+c^{2^m})(c^{2^{m+1}}+c^{2^m}+1)=0.
\end{equation}
Since $c\ne 0, 1$, then we have $c^2+c\ne 0$, $c^{2^m}+c\ne 0$ and $c+c^{2^{m+1}}\ne 0$.
Assume that $c^{2^m}+c=0$, then we have $c\in \gf_{2^m}$. Since ${\rm gcd}(m, 2m+1)=1$, we get $c\in \gf_{2}$, it leads to a contradiction. Assume that $c+c^{2^{m+1}}=0$, then we have $c\in \gf_{2^{m+1}}$. Since ${\rm gcd}(m+1, 2m+1)=1$, we get $c\in \gf_{2}$, which contradicts with $c\ne 0, 1$. Multiplying $c+c^{2^m}$ and $c^2+c$ on both sides of (\ref{FBCT-d=niho-1}) and (\ref{FBCT-d=niho-odd-1}) and then summing up these two equations gives
$$
(c^2+c)(c+c^{2^{m+1}})y^{2^m}+(c+c^{2^m})(c^{2^{m+1}}+c)y^2+(c^2+c^{2^m})(c^{2^{m+1}}+c)y+(c^{2^m}+c^2)(c^{2^m}+c)^2=0,
$$
which can be rewritten as
$$
y^{2^m}=\Big(\frac{(c+c^{2^m})(c^{2^{m+1}}+c)}{(c^2+c)(c+c^{2^{m+1}})}\Big)y^2+\Big(\frac{(c^2+c^{2^m})(c^{2^{m+1}}+c)}{(c^2+c)(c+c^{2^{m+1}})}\Big)y+\frac{(c^{2^m}+c^2)(c^{2^m}+c)^2}{(c^2+c)(c+c^{2^{m+1}})}.
$$
Squaring both sides of the above equation and substituting it into (\ref{FBCT-d=niho-1}) and then multiplying $(c^2+c)(c^2+c^{2^{m+2}})$ on the both sides of this equation, we have
$$
\begin{aligned}
&(c^{2^{m+2}}+c^2)^{2^m+1}y^4+(c^{2^{m+2}}+c^2)^{2^m+1}(c^2+c+1)y^2+(c^{2^{m+2}}+c^2)^{2^m+1}(c^2+c)y\\
&+(c^{2^m}+c^2)^2(c^{2^m}+c)^4+(c^{2^{m+1}}+c)^3(c^4+c)=0.
\end{aligned}
$$
Since $c^{2^{m+1}}+c\ne 0$, then the above equation is equivalent to
\begin{equation}\label{FBCT-d=niho-odd-2}
y^4+(c^2+c+1)y^2+(c^2+c)y+\frac{(c^{2^m}+c^2)^2(c^{2^{m+1}}+c^2)}{c^{2^{m+2}}+c^2}+\frac{(c^{2^{m+1}}+c)(c^4+c)}{c^{2^{m+1}}+c^2}=0.
\end{equation}
If $\frac{(c^{2^m}+c^2)^2(c^{2^{m+1}}+c^2)}{c^{2^{m+2}}+c^2}+\frac{(c^{2^{m+1}}+c)(c^4+c)}{c^{2^{m+1}}+c^2}=0$, then (\ref{FBCT-d=niho-odd-2}) can be reduced to
$$y^4+(c^2+c+1)y^2+(c^2+c)y=y(y+1)(y+c)(y+c+1)=0,$$
this implies that (\ref{FBCT-d=niho-odd-2}) has four solutions in $\gf_{2^n}$, namely, $y=0$, $y=1$, $y=c$ and $y=c+1$, which contradicts with $y\ne 0, 1, c, c+1$. Then we have $\frac{(c^{2^m}+c^2)^2(c^{2^{m+1}}+c^2)}{c^{2^{m+2}}+c^2}+\frac{(c^{2^{m+1}}+c)(c^4+c)}{c^{2^{m+1}}+c^2}\ne 0$. By Lemma \ref{lemma1-quar-root}, the companion cubic polynomial of (\ref{FBCT-d=niho-3}) is
\begin{equation*}
G(z)=z^3+(c^2+c+1)z+c^2+c,
\end{equation*}
which can be factored as $(z+1)(z+c)(z+c+1)$ in $\gf_{2^n}$. If $G(z)=0$, we get $z_1=1$, $z_2=c$ and $z_3=c+1$. Let $a_0=\frac{(c^{2^m}+c^2)^2(c^{2^{m+1}}+c^2)}{c^{2^{m+2}}+c^2}+\frac{(c^{2^{m+1}}+c)(c^4+c)}{c^{2^{m+1}}+c^2}$, $a_1=c^2+c$, $a_2=c^2+c+1$, $\omega_1=\frac{a_0z_1^2}{a_1^2}=\frac{a_0}{a_1^2}$, $\omega_2=\frac{a_0z_2^2}{a_1^2}=\frac{a_0c^2}{a_1^2}$, $\omega_3=\frac{a_0z_3^2}{a_1^2}=\frac{a_0(c+1)^2}{a_1^2}$. Since $G(z)$ can be factored as $(1, 1, 1)$, from Lemma \ref{lemma1-quar-root}, we can easily seen that (\ref{FBCT-d=niho-odd-2}) has four solutions in $\gf_{2^n}$ if and only if  ${\rm Tr_1^n}(\omega_1)={\rm Tr_1^n}(\omega_2)={\rm Tr_1^n}(\omega_3)=0$. This completes the proof.
\end{proof}

\section{The second-order zero differential spectra of functions in odd characteristic}\label{odd-result1}

In this section, we deal with the computation of the second-order zero differential spectra of the function $F(x)=x^5$ over $\gf_{p^n}$ ($p>2$ and $p\ne 5$) and the function $F(x)=x^7$ over $\gf_{p^n}$ ($p>2$ and $p\ne 7$) in Theorem \ref{d=5} and Theorem \ref{d=7}, respectively.

\begin{theorem}\label{d=5}
Let $F(x)=x^5$ be a power mapping over $\gf_{p^n}$, where $p$ is an odd prime and $p\ne 5$. For $a, b \in \gf_{p^n}$, then $F(x)$ is second-order zero differential $3$-uniform. Moreover,
$$
{\nabla}_{F}(a,b)=
\begin{cases}
   p^n, &  {\rm if}\,\, ab=0,\\
   1, &  {\rm if}\,\,  \eta(-(a^2+b^2))=-1; \\
   3, &  {\rm if}\,\, \eta(-(a^2+b^2))=1;\\
\end{cases}
$$
where $\eta$ be the quadratic character of $\mathbb{F}_{p^n}$.
\end{theorem}

\begin{proof}
To prove this theorem, according to Definition \ref{definition-SOZ}, we need to
count the number of the solutions of
$$F(x+a+b)-F(x+b)-F(x+a)+F(x)=0,$$
i.e.,
\begin{equation}\label{diff-d=5}
(x+a+b)^5-(x+b)^5-(x+a)^5+x^5=0,
\end{equation}
where $a,b\in \gf_{p^n}$.

If $ab=0$, then ${\nabla}_{F}(a,b)=p^n$. For $a, b \in \gf_{p^n}^*$, expanding each of the terms of (\ref{diff-d=5}) leads to
$$
20abx^3+30(a^2b+ab^2)x^2+5(4a^3b+6a^2b^2+4ab^3)x+5(a^4b+2a^3b^2+2a^2b^3+ab^4)=0.
$$
Since $ab\ne 0$, the above equation is equivalent to
\begin{equation}\label{diff-d=5-2}
x^3+\frac{3(a+b)}{2}x^2+\frac{2a^2+3ab+2b^2}{2}x+\frac{a^3+2a^2b+2ab^2+b^3}{4}=0.
\end{equation}

When $p=3$, (\ref{diff-d=5-2}) becomes
$$
x^3+(a^2+b^2)x+(a^3-a^2b-ab^2+b^3)=0,
$$
which can be further rewritten as
$$
(x-(a+b))(x^2+(a+b)x-(a-b)^2)=0.
$$
Then we have $x=a+b$ or $x^2+(a+b)x-(a-b)^2=0$.
It can be easily seen that the discriminant of the equation $x^2+(a+b)x-(a-b)^2=0$ is equal to $-(a^2+b^2)$. Then we get (\ref{diff-d=5}) has three solutions in $\gf_{3^n}$ if $\eta(-(a^2+b^2))=1$, and (\ref{diff-d=5}) has exactly one solution in $\gf_{3^n}$ if $\eta(-(a^2+b^2))=-1$.

When $p>5$, let $x=y-\frac{a+b}{2}$, then (\ref{diff-d=5-2}) becomes
\begin{equation*}
y^3+\frac{a^2+b^2}{4}y=y(y^2+\frac{a^2+b^2}{4})=0.
\end{equation*}
One can easily observe that (\ref{diff-d=5}) has three solutions in $\gf_{p^n}$ if $\eta(-(a^2+b^2))=1$, and (\ref{diff-d=5}) has exactly one solution in $\gf_{p^n}$ if $\eta(-(a^2+b^2))=-1$. This completes the proof.
\end{proof}

\begin{theorem}\label{d=7}
Let $F(x)=x^7$ be a power mapping over $\gf_{p^n}$, where $p$ is an odd prime and $p\ne 7$. Let $\eta$ be the quadratic character of $\mathbb{F}_{p^n}$. For $a, b \in \gf_{p^n}$, the second-order zero differential uniformity of $F(x)$ is less than or equal to $5$. In particular, when $p=3$, $F(x)$ is second-order zero differential $3$-uniform. Moreover, for odd $n$,
$$
{\nabla}_{F}(a,b)=
\begin{cases}
   3^n, &  {\rm if}\,\, ab=0;\\
   1, &  {\rm if}\,\,  \eta(\frac{1}{a^2+b^2})=1; \\
   3, &  {\rm if}\,\, \eta(\frac{1}{a^2+b^2})=-1,\\

\end{cases}
$$
and for even $n$,
$$
{\nabla}_{F}(a,b)=
\begin{cases}
   3^n, &  {\rm if}\,\, ab=0;\\
   1, &  {\rm if}\,\,  ab\ne 0,\,\, a\ne b,\,\,  a^2+b^2=0\,\, {\rm or}\,\, ab\ne 0,\,\,  a^2+b^2\ne 0,\,\, \eta(\frac{1}{a^2+b^2})=-1; \\
   3, &  {\rm if}\,\, ab\ne 0,\,\,  a^2+b^2\ne 0,\,\, \eta(\frac{1}{a^2+b^2})=1.\\
\end{cases}
$$
\end{theorem}

\begin{proof}
To prove this theorem, it suffices to consider the number of solutions of
\begin{equation}\label{diff-d=7}
(x+a+b)^7-(x+b)^7-(x+a)^7+x^7=0,
\end{equation}
where $a,b\in \gf_{3^n}$.

If $ab=0$, then ${\nabla}_{F}(a,b)=3^n$. For $a, b \in \gf _{3^n}^*$, expanding each of the terms of (\ref{diff-d=7}) gives
\begin{equation}\label{diff-d=7-1}
\begin{aligned}
&7[6abx^5+15(a^2b+ab^2)x^4+10(2a^3b+3a^2b^2+2ab^3)x^3+\\
&15(a^4b+2a^3b^2+2a^2b^3+ab^4)x^2+(6a^5b+15a^4b^2+20a^3b^3+15a^2b^4+6ab^5)x+\\
&(a^6b+3a^5b^2+5a^4b^3+5a^3b^4+3a^2b^5+ab^6)]=0.
\end{aligned}
\end{equation}
When $p=3$, (\ref{diff-d=7-1}) can be reduced to
$$
(a^3b+ab^3)x^3+a^3b^3x-ab(a+b)(a-b)^4=0.
$$
Since $ab\ne 0$, then the above equation becomes
$$
(a^2+b^2)x^3+a^2b^2x-(a+b)(a-b)^4=0;
$$
which can be further rewritten as
$$
(x-(a+b))((a^2+b^2)x^2+(a+b)(a^2+b^2)x+(a-b)^4)=0,
$$
we have $x=a+b$ or $(a^2+b^2)x^2+(a+b)(a^2+b^2)x+(a-b)^4=0$. For the later quadratic equation, when $n$ is odd, we get $a^2+b^2\ne 0$ since $\eta(-1)=-1$. It can be computed that the discriminant of the quadratic equation is equal to $-\frac{a^2b^2}{a^2+b^2}$, then we have (\ref{diff-d=7}) has three solutions in $\gf_{3^n}$ if $\eta(\frac{1}{a^2+b^2})=-1$, and (\ref{diff-d=7}) has one solution in $\gf_{3^n}$ if $\eta(\frac{1}{a^2+b^2})=1$ when $n$ is odd. When $n$ is even, we have $a^2+b^2$ may be equal to $0$ since $\eta(-1)=1$. Assume that $a^2+b^2=0$, we get $(x-(a+b))(a-b)^4=0$. If $a=b$, then we have $a=b=0$, which contradicts with $a, b\ne 0$. If $a\ne b$, then we get (\ref{diff-d=7}) has exactly one solution, namely, $x=a+b$. When $a^2+b^2\ne 0$, our discussion is the same as when $n$ is odd. This proof is completed.
\end{proof}

When $p>3$ and $p\ne 7$, we can easily seen that the degree of (\ref{diff-d=7-1}) is $5$, thus it at most has five solutions in $\gf_{p^n}$. Therefore, the second-order zero differential uniformity of $F(x) =x^7$ is less than or equal to $5$. Since the involved equation (\ref{diff-d=7-1}) has the degree $5$, it seems difficult to calculate the second-order zero differential spectrum. We leave this as an open problem.

\section{Conclusion}\label{con-remarks}

This paper studied the second-order zero differential spectra of some power functions with low differential uniformity by developing techniques to calculate specific equations over finite fields. It is worth noting that all of these power functions have low second-order zero differential uniformity. In the further work, we will investigate more popular functions with low differential uniformity and determine their second-order zero differential spectra.

\end{document}